\documentclass[lettersize,journal]{IEEEtran}
\usepackage{amsmath,amsfonts}
\usepackage{algorithmic}
\usepackage{array}
\usepackage[caption=false,font=normalsize,labelfont=sf,textfont=sf]{subfig}
\usepackage{textcomp}
\usepackage{stfloats}
\usepackage{url}
\usepackage{verbatim}
\usepackage{graphicx}
\def\BibTeX{{\rm B\kern-.05em{\sc i\kern-.025em b}\kern-.08em
    T\kern-.1667em\lower.7ex\hbox{E}\kern-.125emX}}
\usepackage{balance}
\usepackage{amssymb}
\usepackage{algorithm}
\usepackage{cite}
\usepackage{bm}
\usepackage{cuted}
\usepackage{placeins}

\usepackage[utf8]{inputenc}
\usepackage{lineno} 
\usepackage{cases}
\usepackage{float}  
\usepackage{textcomp}
\usepackage{bbold}
\usepackage{xcolor}
\usepackage{siunitx}

\usepackage{caption}
\usepackage{amsthm}

\newtheorem{Prop}{Proposition}

\begin{document}
\title{Robust Beamforming for Practical RIS-Aided RSMA Systems with Imperfect SIC under Transceiver Hardware Impairments}
\author{Xuejun Cheng$^{\dag}$,
	Qian Zhang$^{\dag}$,~\IEEEmembership{Graduate Student Member,~IEEE,}
	Yunnuo Xu,~\IEEEmembership{Member,~IEEE,} \\
	Zheng Dong,~\IEEEmembership{Member,~IEEE,}
	Ju Liu,~\IEEEmembership{Senior Member,~IEEE,}
	Bruno Clerckx,~\IEEEmembership{Fellow,~IEEE}
	
	\thanks{
		This work was supported in part by the National Natural Science Foundation of China under Grant 62071275; 
		The corresponding authors: Ju Liu; Zheng Dong. E-mail: \{juliu, zhengdong\}@sdu.edu.cn. }
	\thanks{Xuejun Cheng, Qian Zhang, Yunnuo Xu, Zheng Dong, and Ju Liu are with School of Information Science and Engineering, Shandong University, Qingdao, 266237, China. (e-mail: chengxuejun@mail.sdu.edu.cn; qianzhang2021@mail.sdu.edu.cn; yunnuo.xu@sdu.edu.cn; zhengdong@sdu.edu.cn; juliu@sdu.edu.cn.}
	\thanks{Bruno Clerckx is with the Department of Electrical and Electronic Engineering, Imperial College London, SW7 2AZ London, U.K. (e-mail: b.clerckx@imperial.ac.uk).}
	\thanks{$\dag$ Equal contribution.}}

\maketitle

\begin{abstract}
Reconfigurable intelligent surface (RIS)-aided rate-splitting multiple access (RSMA) systems have demonstrated remarkable potential in enhancing spectral efficiency.
However, most existing works rely on ideal hardware, which is unrealistic. In practical deployments, RIS elements suffer from amplitude–phase coupling, where transceivers are subject to hardware impairments (HWI), and successive interference cancellation (SIC) in RSMA networks cannot achieve perfect interference elimination for decoded signals.
To address these limitations, we investigate a robust beamforming design for RIS-aided RSMA systems under practical hardware imperfections.
We first characterize the asymptotic signal-to-noise ratio (SNR) of practical RIS systems when the beamformer is designed based on ideal RIS model, thereby theoretically quantifying the resulting performance degradation.
We then derive a closed-form expression for the distortion noise power induced by transceiver HWI, while also accounting for residual interference due to imperfect SIC.
Building on these insights, we establish a comprehensive system model that jointly incorporates all hardware-induced impairments and formulate a multiuser sum rate maximization problem.
To solve the resulting non-convex optimization problem, we develop an efficient block variable relaxation algorithm. Simulation results verify that the proposed scheme significantly outperforms conventional non-orthogonal multiple access (NOMA) approaches, and achieves superior robustness compared with benchmark schemes neglecting HWI, imperfect SIC, or amplitude–phase coupling.
\end{abstract}

\vspace{-20pt}
\section{Introduction}
Emerging vertical applications like extended reality (XR) will drive the demand for novel enablers in future sixth-generation (6G) wireless communication systems with better coverage and higher spectral efficiency~\cite{li2025exploring}. 
Among the key enabling technologies, reconfigurable intelligent surfaces (RIS) and rate-splitting multiple access (RSMA) have attracted significant attention.
RIS reconfigures the wireless propagation environment by adjusting the phase shift of the signal to significantly enhance communication quality and coverage~\cite{li2025covert}. 
Meanwhile, RSMA serves as an effective interference management technique for 6G wireless networks~\cite{9831440}. 
RSMA splits and encodes user messages into common stream shared by all users and private streams for individual users at the transmitter side, with successive interference cancellation (SIC) executed at the receiver side. 
Therefore, RSMA enables flexible interference handling by both decoding interference and treating it as noise~\cite{22}.
Motivated by these advantages, the interplay RIS-RSMA has been studied under various scenarios~\cite{li2022rate,katwe2024rsma}, including addressing beyond-diagonal RIS, which are more advanced than the diagonal RIS. 

While most of the existing beamforming designs are based on perfect hardware, where each reflecting element maintains unit reflection amplitude. However, practical RIS systems exhibit an inherent coupling between reflection amplitude and phase shift due to hardware constraints. This coupling originates from the physical properties of RIS components and is difficult to eliminate completely in practical implementations~\cite{Abeywickrama2020NonIdealRIS}.
Additionally, practical systems suffer from hardware impairments (HWI), including amplifier nonlinearity and phase noise, as well as imperfect SIC.
Some recent studies have investigated the impact of these imperfections on RSMA or RIS systems.
In~\cite{zhou2021secure}, Zhou {\it et al.} investigated RIS-aided secure communication system with transceiver HWI.

In~\cite{10517628}, Abanto-Leon {\it et al.} investigated the radio resource management design for RSMA while accounting for imperfect SIC.
In~\cite{Abeywickrama2020NonIdealRIS}, Abeywickrama {\it et al.} demonstrated the coupling between the reflection amplitude and phase shift of RIS elements.
In~\cite{zhang2024practical}, Zhang {\it et al.} highlighted that the non-uniform amplitude response in practical RIS systems leads to a mismatch between the optimal phase shifts derived under the ideal RIS model and those required in real scenarios, thereby causing performance degradation. 
However, to the best of the authors' knowledge, robust transmission design for practical RIS-aided RSMA systems under these practical hardware imperfections has not been investigated. 

To enhance the robustness of beamforming, we propose a robust scheme for practical RIS-aided RSMA systems with imperfect SIC under transceiver HWI. 
To better understand the effect of hardware mismatch, we have also conducted rigorous theoretical verification of performance loss when applying beamformers designed under ideal constant-modulus RIS model to practical RIS systems.

Specifically, 1) we have derived the asymptotic signal-to-noise ratio (SNR) of a beamformer designed based on ideal RIS model in practical RIS system, theoretically validating the performance loss (Proposition 1); 
2) we have derived the closed-form HWI distortion noise power with imperfect SIC (Appendix A);
3) we have formulated the sum rate maximization (SRM) problem under imperfect hardware systems and developed a block variable relaxation algorithm to tackle the resulting non-convex optimization problem.

\vspace{-9pt}
\section{Systems Model}
\label{s2}
As shown in Fig.~\ref{fig1}, we consider a downlink RIS-aided RSMA system with practical hardware imperfection where the impact of transceiver HWI, imperfect SIC, and the amplitude-phase coupling in practical RIS on the system performance are considered. 
In the above model, an $M$-antenna base station (BS) serves $K$ single-antenna users with the aid of an $N$-element practical RIS. The direct  BS-users links are obstructed. 
Let $\mathbf{G}\in\mathbb{C}^{N\times M}$ and $\bm f_k\in\mathbb{C}^{N\times 1}$ represent the channel from BS to RIS and from RIS to user-$k$, respectively.
The phase-shift matrix of the RIS is denoted by $\bm \Theta = \mathrm{diag}\left(\bm \phi^\mathrm{H}\right)\in \mathbb{C}^{N \times N}$, where $\bm \phi$ is the reflecting coefficient.
\vspace{-25pt}
\subsection{Practical Amplitude Model}
To characterize the amplitude-phase coupling inherent in practical RIS, we adopt the parallel resonant circuit model referenced in~\cite{Abeywickrama2020NonIdealRIS}, whereby the reflection coefficient of each RIS element is denoted as $\phi_n = \beta_n(\theta_n) e^{j\theta_n}$.  
The reflection amplitude $\beta_n$ of the $n$-th element is expressed as the function of the corresponding phase shift $\theta_n$, i.e.,
\begin{equation} \label{phase_amplitude} 
	\begin{split}
		\beta_n(\theta_n) = \rho \left( {\rm sin}(\theta_n - \delta) + 1 \right)^{\alpha} + \beta_{\min},
	\end{split}
\end{equation}
\noindent where $\rho = (1-\beta_{\min})/2^{\alpha}$. $\beta_{\min}\geq 0$, $\delta\geq 0$, and $\alpha\geq 0$ are the constants related to the specific circuit implementation.

\vspace{-15pt}
\subsection{Signal Model}
\subsubsection{RSMA}
Following 1-layer RSMA scheme~\cite{7555358,mao2018rate}, the BS splits and encodes the transmitted signal into a common stream $s_c$ and private streams $\left\{s_k\right\}_{k=1}^K$. 
Let $\bm s = \left[s_c,s_1,\cdots,s_K\right]^\mathrm{T}$ where $s_c,s_k\sim \mathcal{C} \mathcal{N}\left(0,1\right)$ and $\mathbb{E} \left[\bm s \bm s^\mathrm{H}\right] = \mathbf{I}$. 
The linear precoding vectors for the common and private streams are represented by $\bm w_c\in\mathbb{C}^{M \times 1}$ and $\left\{\bm w_k\right\}_{k=1}^K \in\mathbb{C}^{M \times 1}$, respectively.
Therefore, the signal transmitted by the BS is $\bm x = \bm w_cs_c + \sum_{k=1}^{K}{\bm w_k s_k}+ \bm \eta_t$.
In particular, $\bm \eta_t \sim \mathcal{C} \mathcal{N}\left(\bm 0,m_t \widetilde{\text{diag}}(\bm w_c \bm w_c^\mathrm{H}+\sum_{k=1}^{K}\bm w_k \bm w_k^\mathrm{H})\right)$ is the distortion noise attributable to generalized HWI at the transmitter as validated by both theoretical analysis and field measurements~\cite{schenk2008rf,bjornson2014massive}, where $\widetilde{\text{diag}}$ denotes the diagonal matrix with the same diagonal element; 
$m_t \in (0,1)$ is the ratio of transmit distorted noise power to transmit signal power.

Then, the received signal at user-$k$ can be represented as
\begin{align} \label{y_K}
	y_k = \tilde{y}_k + \eta_{r,k}  = \bm h_k^\mathrm{H} \bm x+ \eta_{r,k} + n_k,
\end{align}

\noindent where $\bm h_k^\mathrm{H} = \bm f_k^\mathrm{H} \bm\Theta \mathbf{G}$, $ n_k \sim \mathcal{C} \mathcal{N}\left(0,\sigma_k^2\right)$ denote the additive white Gaussian noise (AWGN) at user-$k$. 
$\eta_{r,k} \sim \mathcal{C} \mathcal{N}\left(0,m_r\mathbb{E}[|\tilde{y}_k|^2]\right)$ is the distortion noise caused by receiver HWI~\cite{schenk2008rf,bjornson2014massive}, where $m_r \in (0,1)$ is the ratio of distorted noise power to undistorted received signal power. 

\noindent\textit{Remark 1: 
Both $m_t$ and $m_r$ characterize the HWI levels, corresponding to the error vector magnitude (EVM).
In practice, EVM is measurable via vector signal analyzer (VSA) following relevant standards such as 3GPP TS 36.521-1. 
According to the 3GPP LTE standard, typical LTE transceivers that support all standardized modulation schemes require EVM below 0.08~\cite{bjornson2014massive}. Therefore, our simulations consider parameter values within the range [0, 0.08].
}
\begin{figure}[t]
	\centering
	\includegraphics[width=0.4\textwidth]{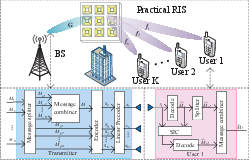}
	\caption{A Practical RIS-aided RSMA system.}
	\label{fig1}
	\vspace{-20pt}
\end{figure}

At the user side, each user initially decodes the common stream $s_c$ by treating all private streams as interference. Subsequently, user-$k$ decodes its private stream after removing $s_c$ via SIC. 
However, different from most existing designs assuming perfect interference cancellation for correctly decoded signals, practical SIC implementation is actually imperfect, which introduces nonnegligible residual self-interference that adversely affects private stream decoding performance.
Therefore, the signal-to-interference-plus-noise ratio (SINR) of user-$k$ under transceiver HWI and imperfect SIC is given as
\vspace{-3pt}
\begin{equation*}
	\gamma_{c,k} \!=\! \frac{|\bm h_k^\mathrm{H} \bm w_c|^2}{\sum_{i=1}^{K}\!{|\bm h_k^\mathrm{H} \bm w_i|^2}\!+\!\Phi_{c,k}}, 
	\gamma_{p,k}\! =\! \frac{|\bm h_k^\mathrm{H} \bm w_k|^2}{\sum_{i=1,i\neq k}^{K}\!{|\bm 	h_k^\mathrm{H} \bm w_i|^2}\!+\!\Phi_{p,k}},
\end{equation*}
\vspace{-8pt}

\noindent where $
\Phi_{c,k} \!=\! \bm h_k^\mathrm{H}[m_r \mathbf{A}\!+\! m_t(1\!+\! m_r)\widetilde{\text{diag}}(\mathbf{A})]\bm h_k \!+\!\sigma^2$,
$\Phi_{p,k} \!\!=\!\!  \delta_{\text{SIC}}^2|\bm h_k^\mathrm{H}\bm w_c|^2\!+\!\Phi_{c,k}, 
\mathbf{A} \!=\! \bm w_c \bm w_c^\mathrm{H}\!+\!\sum_{k=1}^{K}{\!\bm w_k \bm w_k^\mathrm{H}},\sigma^2 \!=\! (1 \! + \! m_r)\sigma_k^2$,
$0 \!\leq\! \delta_{\text{SIC}}\!\!\leq\!\! 1$ denotes the coefficient of SIC imperfection~\cite{10517628}.
The detailed derivation of $\Phi_{c,k}$ and $\Phi_{p,k}$ are provided in Appendix A.
In practical deployments, BS typically cannot obtain the exact value of $\delta_{\text{SIC}}$. To guarantee that the SINR satisfies the requirements for the allocated data rate, this parameter must be configured appropriately. The typical range for $\delta_{\text{SIC}}$ is set between 4\% and 10\%~\cite{9406993}.


The achievable rates of user-$k$ are given as $R_{k} = R_{c,k}+R_{p,k}=\log_2\left(1+\gamma_{c,k}\right) + \log_2\left(1+\gamma_{p,k}\right)$.
To ensure successful common stream $s_c$ decoding by all users, its achievable rate is constrained as $R_{c} \!= \!\min\!{\{R_{c,1},\dots, R_{c, K}\}}$. 
Then, the overall system sum rate is $R_{\text{total}} = R_c +\sum_{k=1}^{K} R_{p,k}$.

\vspace{-0pt}

\subsubsection{NOMA}
The transmitted signal at the BS can be expressed as $\bm x^\text{NOMA}= \sum_{k=1}^{K}{\bm w_k s_k} + \bm \eta_t$.
We assume that user-$k$'s signal is decoded prior to user-$l$'s signal for any $l>k$. Accordingly, the SINR at user-$l$ for decoding user-$k$'s signal and the SINR of user-$K$ are respectively given by
\begin{equation*}
	\gamma_l^k \!\!=\!\! \frac{|\bm h_l^\mathrm{H} \bm w_k|^2}{\sum_{i=k+1}^{K}{\!|\bm h_l^\mathrm{H} \bm w_i|^2}\!\!+\!\Phi_{l,k}}, 
	\gamma_K^K \!\!=\!\! \frac{|\bm h_K^\mathrm{H} \bm w_K|^2}{\sum_{i=1}^{K-1}{\!\delta_{\text{SIC}}^2|\bm h_K^\mathrm{H}\bm w_i|^2\!\!+\!\Phi_{K,K}}}, 
\end{equation*}
\noindent where $\phi_{K,K} \!\!=\!\! \bm h_K^\mathrm{H}[m_r \mathbf{B}\!+ \! m_t(1\!+\! m_r)\widetilde{\text{diag}}(\mathbf{B})]\bm h_K \!+\!(1\!+ \! m_r)\sigma_K^2$,
$\phi_{l,k} \!\!=\!\!\bm h_l^\mathrm{H}[m_r \mathbf{B}\!+ \! m_t(1\!+ \! m_r)\widetilde{\text{diag}}(\mathbf{B})]\bm h_l \!+\!\sigma^2\!+\! \sum_{i=1}^{k}{\!\delta_{\text{SIC}}^2 |\bm h_l^\mathrm{H} \bm w_i|^2}$,
$\mathbf{B} = \sum_{k=1}^{K}\bm w_k \bm w_k^\mathrm{H}$.
$k \in \mathcal{T}_k \triangleq\{1,2,\ldots,K-1\}, l \in \mathcal{T}_l \triangleq\{k,k+1,\ldots,K\}$.
Accordingly, the overall system sum rate is $R_{\text{total}} = \sum_{k=1}^{K} \log_{2}\left(1+\min\limits_{l\in \mathcal{T}_{l}} \gamma_{l}^{k}\right)$.

\section{asymptotic SNR}
\vspace{-3pt}

To characterize the SNR degradation caused by amplitude-phase coupling, we derive the asymptotic SNR of beamformers designed based on ideal RIS model when applied to practical RIS systems.
The channels reduce to $\bm f \in \mathbb{C}^{N \times 1}$ and $\bm g \in \mathbb{C}^{N \times 1}$, both following independent Rayleigh fading, i.e., $\bm f \sim \mathcal{CN}(0, \tau_f^2 \mathbf{I})$ and $\bm g \sim \mathcal{CN}(0, \tau_g^2 \mathbf{I})$.

For ideal RIS systems, the asymptotic SNR is obtained by solving $\max \limits_{\bm \theta} |\bm f^\mathrm{H}\mathrm{diag}\left(\bm g\right)\bm \phi|^2$. 
The optimal phase shift achieves phase alignment of all reflected paths, i.e., $\theta_n^\star \!=\! -\!\arg(f_n^* g_n)$. According to~\cite{zhang2024beyond}, the asymptotic SNR is
\vspace{-3.5pt}
\begin{equation*}
	\gamma_I^a \!=\! {\mathbb{E} \!\left[ \left|\sum_{n=1}^{N}\! |f_n| |g_n|\right|^2\right]\! }\frac{P_{\max}}{\sigma^2} \!=\!\frac{\tau_f^2\tau_g^2 P_{\max}}{16\sigma^2}\!\left[N^2\pi^2\! +\! N\!\left(16\!-\!\pi^2\right) \right]\!,
	\vspace{-3.5pt}
\end{equation*}
\noindent where $P_{\max}$ is the BS power budget, and $\sigma^2$ is the noise power.
Subsequently, the optimal phase shifts are applied to practical RIS systems. 
Therefore, the asymptotic SNR of the practical RIS is altered and expressed as $\gamma^a_P = \mathbb{E} \left[ \left|\sum_{n=1}^{N} |f_n| \beta_n\left(\theta_n^\star \right) |g_n|\right|^2\right] \! \frac{P_{\max}}{\sigma^2}$, bounded by $\left[\beta_{\min}^2 \gamma_I^a, \gamma_I^a\right]$.
\vspace{-8pt}
\begin{Prop}
	\label{proposition_1}
	The ratio of the asymptotic SNR under practical RIS model and ideal RIS model is given as follows
	\vspace{-3pt}
	\begin{equation}
		\label{eta}
		\eta \! = \! \frac{\gamma^a_P}{\gamma^a_I} \! = \!	\left[\frac{1 \!-\! \beta_{\min}}{2^\alpha}\left( 1 \!+\! \sum_{k=1}^{\lfloor S/2 \rfloor} \frac{\prod_{t=0}^{2k \!-\!1}(\alpha \! - \! t)}{2^{2k} (k!)^2}\right) \!+\! \beta_{\min}\right]^2 \!\leq\! 1,
		\vspace{-2pt}
	\end{equation}
	\noindent where $S$ is the truncation order of the Taylor expansion, $t \geq 1$ is a positive integer, and $k!$ is the factorial of $k$.
\end{Prop}
\vspace{-6pt}
{\it Proof:} See Appendix B.

\noindent\textit{Remark 2: Proposition 1 is illustrated using independent Rayleigh fading channels. However, the conclusion remains valid when spatial correlation exists, i.e., under correlated Rayleigh fading channel models.
}


Based on \eqref{eta}, we theoretically demonstrate that phase shift designs based on ideal RIS assumptions inevitably lead to performance losses in practical systems. 
Furthermore, for RIS-aided multi-antenna communications, neglecting the amplitude-phase coupling effect results in more severe performance degradation. 
Therefore, accounting for the amplitude-phase coupling effect is essential for enhancing system performance in the design of RIS-aided MIMO systems.

\vspace{-7pt}
\section{Problem Formulation and Beamforming Optimization}
\vspace{-5pt}
\label{s3}
In this section, we aim to maximize the multiuser sum rate in practical RIS-aided RSMA systems to enhance the robustness of the beamformer design. 
Notably, while we adopt the amplitude-phase model in \eqref{phase_amplitude}, the proposed method is applicable to other amplitude-phase coupling models. The SRM problem is formulated as
\vspace{-5pt}
\begin{subequations}
	\label{problem_formulation}
	\begin{align}
		\max\limits_{\bm w,\bm \Theta} \ & f\!\left(\!\bm w,\bm \Theta\!\right)\! = \log_2\!\left(\!1 \!+\! \min\limits_{k \in \mathcal{K}}{\gamma_{c,k}}\!\right)\! \!+\! \sum_{k=1}^{K} \log_2\!\left(1 \!+\! \gamma_{k}\right)\!\\
		\text{s.t.} \ & \mathcal{C}_{\text{BS}}: \|\bm{w}_c\|_2^2 + \sum_{i=1}^{K}{\|\bm{w}_i\|_2^2} \leq P_{\max}, \\[-2pt]
		& \mathcal{C}_k : \min\limits_{k \in \mathcal{K}}{\gamma_{c,k}}/K + \gamma_{k} \geq \gamma_{\text{th}}, \ \mathcal{C}_{\theta}: -\pi \leq \theta_n \leq \pi, \\[-2pt]
		&\mathcal{C}_{\beta}\!: \beta_n(\theta_n) = \rho \left( {\rm sin}(\theta_n - \delta) + 1 \right)^{\alpha} + \beta_{\min},
	\end{align}
\end{subequations}
\noindent where $\bm{w} = \left[ \bm{w}_c^\mathrm{T},\bm{w}_1^\mathrm{T},\bm{w}_2^\mathrm{T},\dots,\bm{w}_K^\mathrm{T} \right]^\mathrm{T}$ denotes the jointly optimized common and private beamforming vectors.
$\gamma_{\text{th}}$ represents the quality of service (QoS) threshold.

To convexify the objective function, we propose an effective variable relaxation algorithm.
Specifically, we introduce slack variables $\xi, \xi_{k}$ and then transform problem \eqref{problem_formulation} equivalently as
\vspace{-13pt}
\begin{subequations}
	\label{problem_formulation1}
	\begin{align}
		\max\limits_{\bm w,\xi,\{\xi_k\}_{k=1}^K, \bm \Theta} \  & \mathcal{Q} = \log_2\left(1 + {\xi}\right) + \sum_{k=1}^{K} \log_2\left(1 + \xi_k\right)  \\
		\text{s.t.} \ &\mathcal{C}_{\text{BS}},\mathcal{C}_{\theta},\mathcal{C}_{\beta}, \mathcal{C}_k : \gamma_{c,k} \geq \xi ,\gamma_{k} \geq \xi_k, \\
		& \mathcal{C}_{\text{th}} : \xi/K + \xi_k \geq \gamma_{\text{th}}, 
		\vspace{-5pt}
	\end{align}
\end{subequations}
\noindent where constraint $\mathcal{C}_{\text{th}}$ ensures that the auxiliary variables $\xi$ and $\xi_k$ satisfy the minimum QoS thresholds.

Due to the coupling between $\bm w$ and $\mathbf \Theta$, the non-convex problem \eqref{problem_formulation1} is significantly complicated. 
Therefore, we decompose it into two subproblems solved iteratively by the alternating optimization (AO) algorithm, which is updated as 
\begin{subequations}
	\label{AO}
	\begin{align}
		\label{AO_w}
		\bm w^{m+1} = \arg \max\limits_{\bm w,\xi,\{\xi_k\}_{k=1}^K} \mathcal{Q}\quad &\text{s.t.} \ \mathcal{C}_{\text{BS}},\mathcal{C}_k,\mathcal{C}_{\text{th}},  \\
		\label{AO_theta}
		\mathbf{\Theta}^{m+1} = \arg \max\limits_{\bm \Theta , \xi,\{\xi_k\}_{k=1}^K} \mathcal{Q}\quad &\text{s.t.} \ \mathcal{C}_k,\mathcal{C}_{\text{th}},\mathcal{C}_{\theta},\mathcal{C}_{\beta}. 
	\end{align}
\end{subequations}

\vspace{-15pt}
\subsection{Optimize Precoding Vector}
In this subsection, problem \eqref{AO_w} can be given by
\begin{equation}
	\label{Precoding Vector}
	\begin{aligned}
		\max\limits_{\bm w,\xi,\{\xi_k\}_{k=1}^K} \  \mathcal{Q} \quad \text{s.t.} \ 	\mathcal{C}_{\text{BS}},\mathcal{C}_k,\mathcal{C}_{\text{th}}.
	\end{aligned}
\end{equation}

Since constraints $\mathcal{C}_k$ are non-convex concerning $\bm w$, problem \eqref{Precoding Vector} cannot be solved directly. 
Therefore, we employ the quadratic transformation method~\cite{shen2018fractional} to reformulate $\mathcal{C}_k$ as
\vspace{-3pt}
\begin{equation*}
	\label{FP}
	\begin{aligned}
		\xi \leq \max\limits_{\mu_{c,k}}\ 2\mathcal{R}\{\mu_{c,k}^{*} \bm h_k^\mathrm{H} \bm w_c\} \!-\! |\mu_{c,k}|^2 (\sum_{i=1}^{K}{|\bm h_k^\mathrm{H} \bm w_i|^2} \!+\! \Phi_{c,k}),\\
		\xi_k \leq \max\limits_{\mu_{k}}\ 2\mathcal{R}\{\mu_{k}^{*} \bm h_k^\mathrm{H} \bm w_k\} \!-\! |\mu_{k}|^2 (\sum_{i=1,i \neq k}^{K}{|\bm h_k^\mathrm{H} \bm w_i|^2}\!+\!\Phi_{p,k}),
	\end{aligned}
\end{equation*}

\noindent which guarantees the same Karush–Kuhn–Tucker (KKT)~\cite{boyd2004convex} points as the original problem \eqref{Precoding Vector}.
We derive the closed-form optimal solutions of $\mu_{c,k},\mu_{k}$ as
\begin{equation}
	\begin{aligned}
		\label{mu}
		\mu_{c,k}^{\star} \!=\! \frac{\bm h_k^\mathrm{H} \bm w_c}{\sum_{i=1}^{K} |\bm h_k^\mathrm{H} \bm w_i|^2 \!+\! \Phi_{c,k}},\  \mu_{k}^{\star} \!=\! \frac{\bm h_k^\mathrm{H} \bm w_k}{\sum_{i=1, i \neq k}^{K} |\bm h_k^\mathrm{H} \bm w_i|^2 \!+\! \Phi_{p,k}}.
	\end{aligned}
\end{equation}

Therefore, problem \eqref{Precoding Vector} is transformed into a convex optimization problem that can be efficiently solved.

\vspace{-10pt}
\subsection{Optimize Phase-Shift Matrix}

In this subsection, problem \eqref{AO_theta} can be given by
\vspace{-3pt}
\begin{equation}
	\label{Phase-Shift Matrix1}
	\begin{aligned}
		\max\limits_{\bm \Theta,\xi,\{\xi_k\}_{k=1}^K} \ \mathcal{Q} \quad \text{s.t.} \  \mathcal{C}_{\text{th}}, \mathcal{C}_k, \mathcal{C}_{\theta}, \mathcal{C}_{\beta}.
	\end{aligned}
\end{equation}

However, the amplitude-phase coupling results in the highly complicated and non-convex problem \eqref{Phase-Shift Matrix1}.
Therefore, we introduce an auxiliary constraint $\tilde{\bm \phi} = \bm \phi$ to decouple the complicated constraints and reformulate problem \eqref{Phase-Shift Matrix1} using the augmented Lagrangian method~\cite{dai2023augmented} in scaled form as
\vspace{-3pt}
\begin{equation}
	\label{Phase-Shift Matrix2}
	\begin{aligned}
		\max\limits_{\bm\phi,\xi,\xi_k}\ &\mathcal{Q}-\lambda{\|{\tilde{\bm \phi}-\bm \phi+\bm \nu}\|_2^2} \quad \text{s.t.} \ \mathcal{C}_{\text{th}}, \mathcal{C}_k, \mathcal{C}_{\theta}, \mathcal{C}_{\beta},
	\end{aligned}
	\vspace{-5pt}
\end{equation}

\noindent where $\lambda > 0$ is a given parameter, and $\bm \nu \in \mathbb{C}^{N \times 1}$ is a scaled dual variable.

Following the principles of the alternating direction method of multipliers (ADMM)~\cite{boyd2011distributed}, problem \eqref{Phase-Shift Matrix2} can be solved through the following update method.
\vspace{-3pt}
\begin{subequations}
	\begin{flalign}
		\label{ADMM_theta}
		\bm \phi^{l+1}  =& \arg\max\limits_{\bm \phi}\ \mathcal{Q}-\lambda{\|{\tilde{\bm\phi}^l-\bm\phi+\bm\nu^l}\|_2^2}, \\
		\text{s.t.}\  &\xi/K + \xi_k \geq \gamma_{\text{th}}, \gamma_{c,k} \geq \xi ,\gamma_{k} \geq \xi_k, \notag\\
		\label{ADMM_hat_theta}
		\bm{\tilde{\bm  \phi}}^{l+1}=&\arg\min\limits_{\tilde{\bm \phi}} \ {\|{ \tilde{\bm\phi}-\bm\phi^{l+1}+\bm\nu^l}\|_2^2}, \\	
		\text{s.t.} \  &\!-\!\pi \!\leq\! {\theta}_n \!\leq\! \pi,	\beta_n({\theta}_n)\! =\! \rho \left( {\rm sin}({\theta}_n \!-\! \delta) \!+\! 1 \right)^{\alpha} \!+\! \beta_{\min},\notag\\
		\label{ADMM_nu}
		\bm\nu^{l+1}=&\bm\nu^l+\tilde{\bm \phi}^{l+1}-\bm\phi^{l+1}. 
	\end{flalign}
	\vspace{-20pt}
\end{subequations}

For non-convex constraints in subproblem \eqref{ADMM_theta}, we adopt a similar relaxation method as in problem \eqref{FP}. Define
\vspace{-3pt}
\begin{align}
	&\bm{\Lambda}_k \!=\! \mathbf{G}^\mathrm{H}\mathrm{diag}\left(\bm f_k\right), \mathbf{P}_{c,k} \!=\! \bm{\Lambda}_k^\mathrm{H} \mathbf{B}_{c,k} \bm{\Lambda}_k, \mathbf{P}_{k} \!=\! \bm{\Lambda}_k^\mathrm{H} \mathbf{B}_{k}\bm{\Lambda}_k,\notag \\
	&\mathbf{B}_{c,k} \!=\! m_r\mathbf{A}\! + \!m_t(1\! + \! m_r)\widetilde{\text{diag}}(\mathbf{A})\! + \!\textstyle\sum_{i=1}^{K}{\bm w_i \bm w_i^\mathrm{H}}, \notag \\
	&\mathbf{B}_{k} \!=\! m_r\mathbf{A}\! + \!m_t(1\! + \! m_r)\widetilde{\text{diag}}(\mathbf{A})\! + \! \delta_{\text{SIC}}^2 \bm w_c \bm w_c^\mathrm{H}\! + \!\textstyle\sum_{i=1,i \neq k}^{K}{\bm w_i \bm w_i^\mathrm{H}}, \notag\\
	&h\! \left(\mu_{c, k},\bm \phi\right) \!=\! 2\mathcal{R}\!\left\{\mu_{c,k}^*\bm w_c^\mathrm{H}\bm\Lambda_{k}\bm\phi\right\}\!-\!{|\mu_{c,k}|^2}\!\left(\bm\phi^\mathrm{H}\mathbf{P}_{c,k}\bm\phi \! + \! \sigma^2\right)\!, \notag\\
	&h\! \left(\mu_{k},\bm \phi\right) \!=\! 2\mathcal{R}\!\left\{\mu_{k}^*\bm w_k^\mathrm{H}\bm\Lambda_{k}\bm\phi\right\} \! - \! {|\mu_{k}|^2}\!\left(\bm\phi^\mathrm{H}\mathbf{P}_{k}\bm\phi\!+\!\sigma^2\right)\!,\notag
	\vspace{-5pt}
\end{align}
\noindent where $\mu_{c, k},\!\mu_{k}$ are calculated based on \eqref{mu}. Constraints $\gamma_{c,k} \!\geq\! \xi,\! \gamma_{k} \!\geq\! \xi_k$ are equivalent to 
$h\left(\mu_{c, k},\bm \phi\right) \!\geq\! \xi,\!\ h\left(\mu_{k},\bm \phi\right) \!\geq\! \xi_{k}. $
Now, subproblem \eqref{ADMM_theta} has been successfully transformed into a convex problem.

Furthermore, we propose a concise and effective method to solve problem \eqref{ADMM_hat_theta}. The problem \eqref{ADMM_hat_theta} can be tackled by 
\vspace{-3pt}
\begin{subequations}
	\label{line_search_method1}
	\begin{align}
		\min_{ \{ \zeta_n \}_{n=1}^N }\ &\sum\nolimits_{n=1}^{N} \left| \zeta_n - \left(\theta_n^{\ell+1} -\nu_n^{\ell}\right) \right|^2 \\
		\text{s.t.}\ 
		& -\pi \leq \theta_n \leq \pi,\  \zeta_n = \beta_n\left( \theta_n\right) e^{j \theta_n},\\
		&\beta_n(\theta_n ) = \rho \left( {\rm sin} (\theta_n - \delta) + 1 \right)^{\alpha} + \beta_{\min}.
	\end{align}
	\vspace{-17pt}
\end{subequations}

Then, the optimal solution of the problem \eqref{line_search_method1} can be obtained by the Proposition \ref{proposition_2}. 
\vspace{-3pt}
\begin{Prop}
	\label{proposition_2}
	Let $\vartheta_n \!= \!\theta_n^{\ell+1} \!-\!\nu_n^{\ell}, f(\theta_n|\vartheta_n)\!=\!   |\vartheta_n|^2 \!+\! \beta_{\min}^2\\ \!+\! \xi^2\! \left({\rm sin}(\theta_n \! -\! \delta) \!+\! 1\right)^{2\alpha} \!+\! 2\xi ({\rm sin}(\theta_n \!-\! \delta) \!+\! 1)^{\alpha} [ \beta_{\min} \!-\! {\rm Re}\{ \vartheta_n \} {\rm cos}(\theta_n) \!-\! {\rm Im}\{\vartheta_n\}{\rm sin}(\theta_n) ]   \!-\! 2\beta_{\min} [{\rm Re}\{ \vartheta_n \} {\rm cos}(\theta_n) \!+\! {\rm Im}\{\vartheta_n\}{\rm sin}(\theta_n) ]$, we can obtain that the value of $\theta_n$ maximinzing $f(\theta_n|\vartheta_n)$ is the global optimal solution to problem~\eqref{line_search_method1}.
	\vspace{-15pt}
\end{Prop}
{\it Proof:} See Appendix C.

\noindent\textit{Remark 3: 
	The objective function value of the problem \eqref{AO} keeps increasing with the iteration process by employing the AO algorithm, which satisfies $f\left(\bm w^{m+1},\mathbf{\Theta}^{m+1}\right) \!\geq\! f\left(\bm w^{m+1},\mathbf{\Theta}^{m}\right)\!\geq\! f\left(\bm w^{m},\mathbf{\Theta}^{m}\right)$.
	Furthermore, due to the limits on transmit power and the number of antennas, the multiuser sum rate has an upper bound. Therefore, we can prove that the proposed algorithm converges to a locally optimal solution.}

Notably, the proposed algorithm can directly solve the SRM problem in practical RIS-aided NOMA or SDMA systems.
In addition, we establish that, as $\delta_{\mathrm{SIC}}\to 1$, the optimal RSMA structure reduces to SDMA (i.e., $\bm w_c^\star=\bm 0$), with proof in~\cite{2026arXiv260218077C}.

\vspace{-9pt}
\section{Simulation Results}
\vspace{-3pt}
\label{s4}

In this section, we demonstrate the effectiveness of our method through numerical results.
We set $\tau_f^2\!=\!\tau_g^2\!=\!1$, $S \!=\!5$, and $P_{\max}/\sigma^2 \!=\! 1$. 
Fig.~\ref{plot_Asymptotic_SNR} illustrates the variation of asymptotic SNR versus the number of RIS elements.
The simulation results from $10^6$ Monte Carlo trials validate the accuracy of the proposed theoretical expression for the asymptotic SNR in practical RIS systems. 
Notably, the correlation curves demonstrate that our theoretical analysis remains applicable to Rayleigh correlated channels.
The results simply reveal that applying phase shifts derived from ideal RIS to practical RIS systems leads to significant performance loss. 
As $\beta_{\min}$ increases, the performance gradually improves in the theoretical region.
However, even for a larger value of $\beta_{\min}$, the performance loss remains non-negligible, demonstrating that the amplitude-phase coupling must be considered in RIS phase shift design. 
Furthermore, all curves maintain the quadratic growth with the number of RIS elements $N$, confirming that practical RIS can still achieve the square law of SNR. 
\begin{figure}[tbp]
	\centering
	\includegraphics[width=0.3\textwidth]{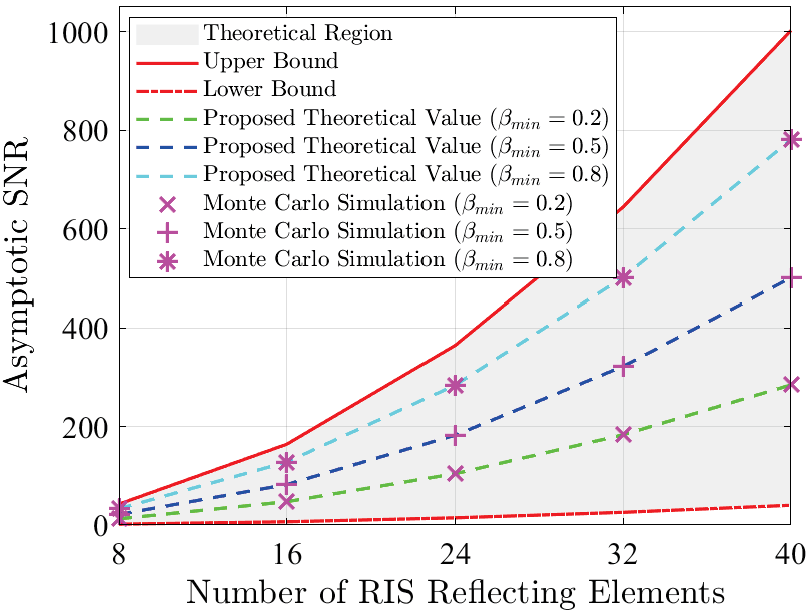}
	\caption{Asymptotic SNR versus number of the RIS elements.}
	\label{plot_Asymptotic_SNR}
	\vspace{-15pt}
\end{figure}

In our simulation, the locations of the BS and RIS are set to $(0, 0, 0)$ and $(40, 40, 0)$ meters. 
The users are uniformly distributed within an 8-meter radius circle centered at $(50, 0, 0)$ meters on the $x \!-\! y$ plane.
The large-scale fading for the BS-RIS and RIS-user channels is modeled by $\mathrm{PL} \!=\! 37.3 \!+ \!22.0\log\left(d\right)$, where $d$ denotes the distance between devices.
The small-scale fading follows the Rician channel model~\cite{wang2021beamforming}. 
Furthermore, to more comprehensively evaluate the robustness of proposed framework under different multipath channel, we further consider a line-of-sight (LoS)-dominant scenario with extremely weak scattering.
We set $K \!=\! 2, M \!=\!8, N \!=\!16, \sigma_k^2 \!=\! -110$ dBm, $\gamma_{\text{th}} \!=\!1$, $\beta_{\min} \!=\! 0.2$, $\delta\!=\!0.43\pi$, and $\alpha\!=\!1.6$. 
All simulation results are averaged over 200 independent trials for statistical reliability.
Simulation results validate the favorable convergence performance of the proposed algorithm under various configurations.
Fig.~\ref{plot_convergence} demonstrates the convergence of the proposed algorithm for practical RIS-aided RSMA systems. As observed, the proposed scheme achieves superior performance compared to the ideal RIS scheme that neglects amplitude-phase coupling effects.
\begin{figure}[tbp]
	\centering
	\includegraphics[width=0.29\textwidth]{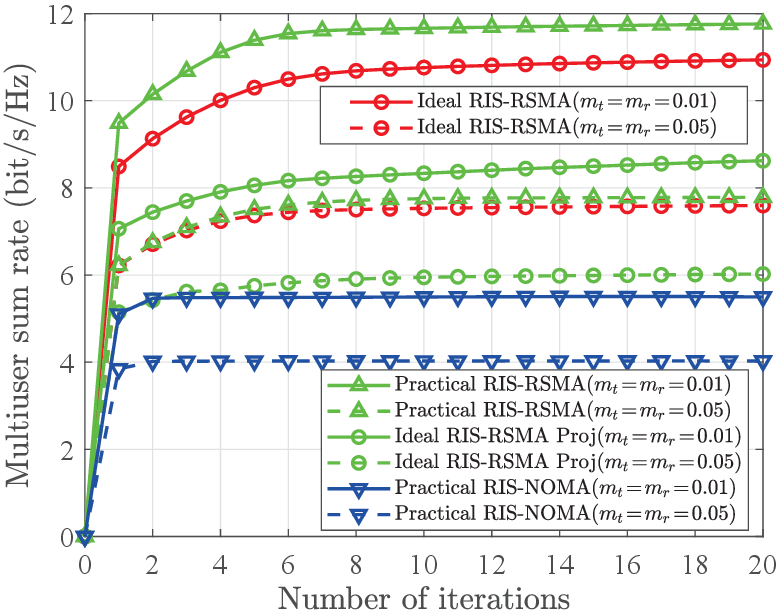}
	\caption{Convergence of the proposed AO algorithm.}
	\label{plot_convergence}
	\vspace{-15pt}
\end{figure}

\begin{figure}[tbp]
	\centering
	\includegraphics[width=0.29\textwidth]{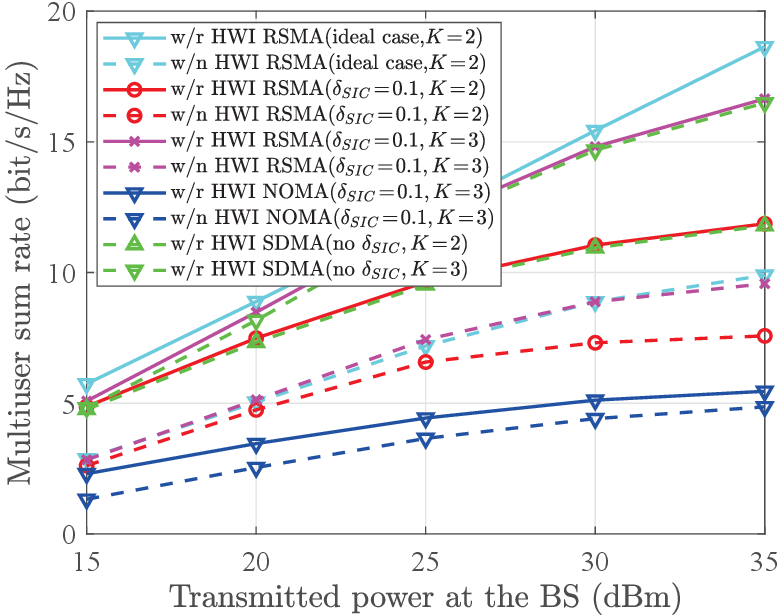}
	\caption{The sum rate versus transmitted power under different multiple access systems with $m_t=m_r=0.01$.}
	\label{HWI}
	\vspace{-20pt}
\end{figure}

To validate the robustness of the proposed scheme, Fig.~\ref{HWI} compares the multiuser sum rate performance across different beamforming designs and multiple access techniques.  Specifically, the following schemes are considered:
1) ``w/r", robust beamforming design that jointly considers HWI, imperfect SIC, and amplitude-phase coupling of the RIS.
2) ``w/n HWI", non-robust beamforming design that considers only HWI.
The results show the proposed robust RSMA scheme (w/r) consistently outperforms the non-robust (w/n) performance under the same channel, thereby validating the necessity of jointly incorporating multiple imperfect hardware considerations into system design. 
The proposed RSMA scheme significantly outperforms NOMA systems and achieves performance comparable to SDMA systems, validating the effectiveness of the proposed robust design framework.
Moreover, to ensure fair comparison, all MA schemes are evaluated under identical channel conditions, transmit power budget, and algorithm update rules, with resource allocation parameters independently optimized at each SNR point. 

Fig.~\ref{Diff_parameters_RSMA_R_c_0} illustrates the impact of different parameters on sum rate and the convergence behavior of RSMA. 
In Fig.~\ref{Diff_parameters_RSMA_R_c_0}\subref{Diff_parameters}, the blue, red, and green curves depict system performance versus $m_t \!=\! m_r$, $\delta_{\mathrm{SIC}}$, and $\beta_{\mathrm{mim}}$, respectively.
The ``w/r HWI" curves indicate that increasing $m_t$ and $m_r$ leads to monotonic degradation in system performance.
The ``w/r RIS" curves maintain stable performance across different $\beta_{\min}$, validating that the proposed algorithm accommodates various degrees of RIS amplitude-phase coupling.
Furthermore, the ``w/r SIC" curves show that the system sum rate remains nearly unchanged as $\delta_{\mathrm{SIC}}$ increases, indicating that, under the proposed algorithm, RSMA can adaptively degenerate into SDMA according to the channel conditions and thus becomes insensitive to imperfect SIC. 
This phenomenon is further verified in Fig.~\ref{Diff_parameters_RSMA_R_c_0}\subref{RSMA_R_c_0}, where the common stream rate gradually converges to zero, and the system sum rate gradually approaches that of SDMA.
Fig.~\ref{RSMA_SDMA} shows that RSMA outperforms SDMA under low SIC imperfection in the LoS-dominant multipath channel, with performance gradually degrading as $\delta_{\text{SIC}}$ increases and eventually converging to SDMA, again validating the flexible interference management of RSMA.

\begin{figure}[tbp]
	\centering
	\subfloat[\small Parameter Sensitivity]{
		\includegraphics[width=0.49\linewidth]{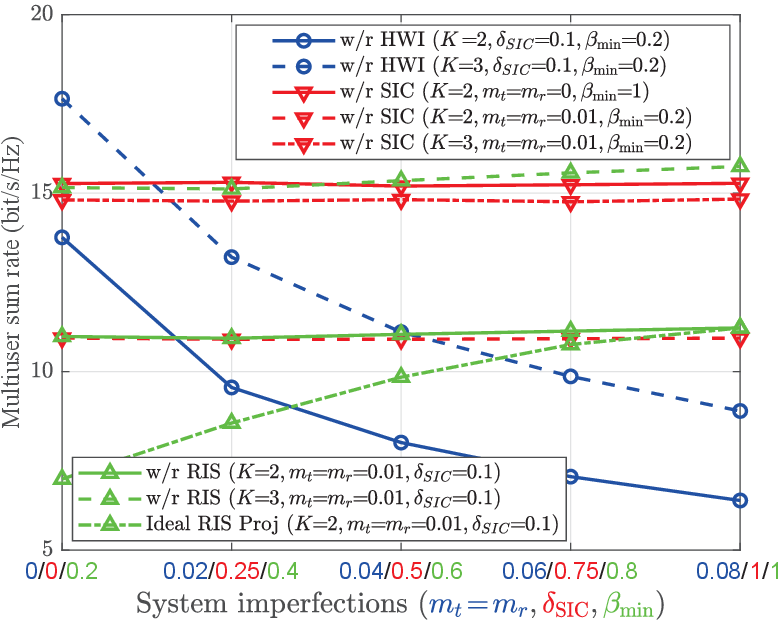}
		\label{Diff_parameters}}
	\subfloat[\small RSMA-SDMA Convergence]{
		\includegraphics[width=0.49\linewidth]{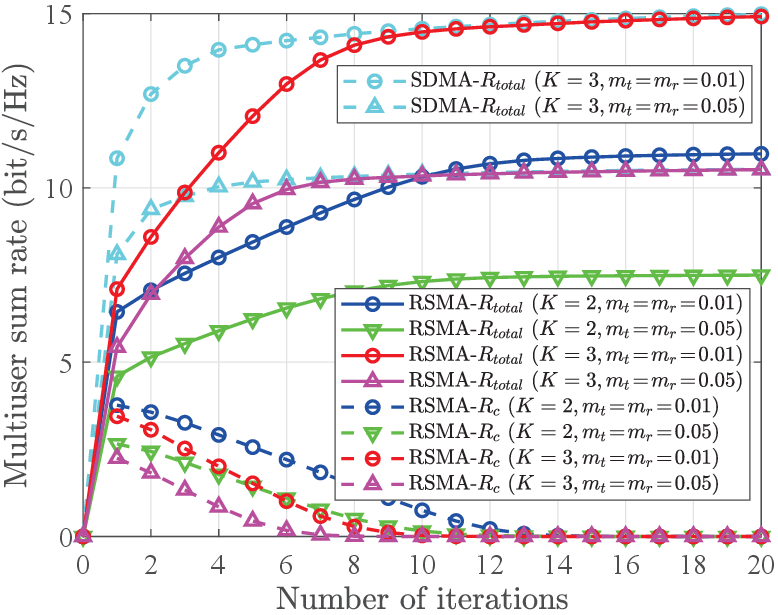}
		\label{RSMA_R_c_0}}%
	\caption{The sum rate versus $m_t, m_r, \delta_{\text{SIC}}$ or $\beta_{\min}$ for different designs and RSMA convergence characteristics.}
	\label{Diff_parameters_RSMA_R_c_0}
	\vspace{-12pt}
\end{figure}

\begin{figure}[tbp]
	\centering
	\includegraphics[width=0.29\textwidth]{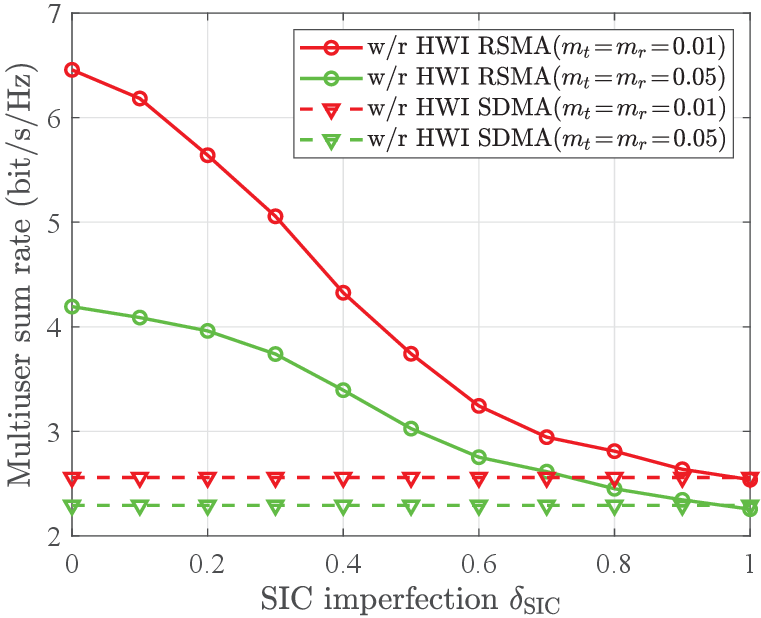}
	\caption{Case study: performance comparison between RSMA and SDMA.}
	\label{RSMA_SDMA}
	\vspace{-20pt}
\end{figure}

\vspace{-10pt}
\section{conclusion}
\label{s5}
This paper investigated robust beamforming design for practical RIS-aided RSMA systems under hardware imperfections, including amplitude-phase coupling, HWI, and imperfect SIC. 
Specifically, we first theoretically characterized the performance loss when applying beamformers designed under the ideal RIS model to practical RIS systems. 
Furthermore, we derived a closed-form expression for the distortion noise power caused by transceiver HWI to quantify the HWI effects and formulated the SRM problems under imperfect SIC.
To solve the resulting non-convex optimization problem, we developed an effective block variable relaxation algorithm.
Simulation results demonstrated that the proposed robust scheme significantly outperforms traditional OMA and NOMA systems and benchmark schemes that neglect hardware imperfections, validating the importance of considering practical hardware limitations in RIS-aided RSMA system design.

\vspace{-12pt}
\section*{Appendix A: Proof of $\Phi_{c,k}$ and $\Phi_{p,k}$}

\label{appendices_0}
According to \eqref{y_K}, we can obtain $\Phi_{c,k}$ and $\Phi_{p,k}$ as follows
\begin{equation*}
	\begin{split}
		\Phi_{c,k} &= \mathbb{E}{[|\bm{h}_k^\mathrm{H}\bm \eta_t|^2]}+\mathbb{E}{[|n_k|^2] +\mathbb{E}{[|\bm \eta_{r,k}|^2}]},\\
		\Phi_{p,k} &= \delta_{\text{SIC}}^2|\bm h_k^\mathrm{H}\bm w_c|^2+\mathbb{E}{[|\bm{h}_k^\mathrm{H}\bm \eta_t|^2]}+\mathbb{E}{[|n_k|^2] +\mathbb{E}{[|\bm \eta_{r,k}|^2}]},
	\end{split}
\end{equation*}
\noindent where each term represents the transmitter HWI, AWGN, and receiver HWI, respectively. 
Then, we have
\begin{equation*}
	\begin{split}
		\mathbb{E}{[|\bm{h}_k^\mathrm{H}\bm \eta_t|^2]} &= \bm{h}_k^\mathrm{H} m_t \widetilde{\text{diag}}(\bm w_c \bm w_c^\mathrm{H}+\textstyle\sum_{k=1}^{K}\bm w_k \bm w_k^\mathrm{H}) \bm h_k,\\
		\mathbb{E}{[|n_k|^2]} &= \sigma_k^2, \mathbb{E}{[|\bm \eta_{r,k}|^2]}  = m_r [\bm{h}_k^\mathrm{H} \mathbb{E}{[\bm x \bm x^\mathrm{H}]} \bm{h}_k+\sigma_k^2].\\
	\end{split}
\end{equation*}
\vspace{-5pt}

Due to $\bm x = \bm w_cs_c + \sum_{k=1}^{K}{\bm w_k s_k}+ \bm \eta_t$, $\mathbf{A} \!=\! \bm w_c \bm w_c^\mathrm{H}\!+\!\sum_{k=1}^{K}\bm w_k \bm w_k^\mathrm{H}, \sigma^2 \!=\! (1 \! + \! m_r)\sigma_k^2$, we can obtain
	\vspace{-5pt}
\begin{equation*}
	\begin{split}
		\mathbb{E}{[\bm x \bm x^\mathrm{H}]} \!=\! \mathbf{A} + m_t \widetilde{\text{diag}}(\mathbf{A}).
	\end{split}
\end{equation*}
\vspace{-15pt}

Therefore, the expression of $\Phi_{c,k}$, $\Phi_{p,k}$ can be obtained as
\begin{equation*}
	\begin{split}
		\Phi_{c,k} &= \bm h_k^\mathrm{H}[m_r \mathbf{A}+m_t(1+m_r)\widetilde{\text{diag}}(\mathbf{A})]\bm h_k+\sigma^2,	\\
		\Phi_{p,k} &= \delta_{\text{SIC}}^2|\bm h_k^\mathrm{H}\bm w_c|^2\!+\!\bm h_k^\mathrm{H}[m_r \mathbf{A}\!+\!m_t(1\!+\!m_r)\widetilde{\text{diag}}(\mathbf{A})]\bm h_k\!+\!\sigma^2,
	\end{split}
\end{equation*}
\vspace{-20pt}

\section*{Appendix B: Proof of Proposition \ref{proposition_1}}
\label{appendices_1}
For asymptotic SNR $\gamma^a_P$, we can obtain 
\vspace{-3pt}
\begin{equation*}
	\begin{aligned}
		\mathbb{E} \left[\left|\sum_{n=1}^{N} |f_n| \beta_n\left(\theta_n^\star\right)|g_n|\right|^2 \right] =  \sum_{n=1}^{N}\mathbb{E}[|f_n|^2]\mathbb{E}[|g_n|^2] \mathbb{E}[\beta_n^2]\\
		+\sum_{n=1}^{N}\sum_{m\neq 	n}^{N}\mathbb{E}[|f_n|]\mathbb{E}[|f_m|]\mathbb{E}[|g_n|]\mathbb{E}[|g_m|\mathbb{E}[\beta_n]\mathbb{E}[\beta_m]]. \\
	\end{aligned}
	\vspace{-3pt}
\end{equation*}

Due to the lack of closed-form solutions for $\beta_n\left(\theta_n\right)$ involving nonlinear trigonometric functions and arbitrary power $\alpha$, we can approximately transform $\beta_n\left(\theta_n\right)$ as 
\begin{equation*}
	\begin{split}
		\beta_n(\theta_n) &\approx \frac{1 - \beta_{\min}}{2^\alpha} \left( 1 + \sum_{s=1}^{S} \frac{\prod_{t=0}^{s-1}(\alpha - t)}{s!}\frac{1}{(2j)^s} \right. \\
		&\quad \left. \times \sum_{l=0}^{s} \frac{s!}{l!(s-l)!}(-1)^l e^{j(2l-s)\delta} e^{j(s-2l)\theta_n} \right) + \beta_{\min},
	\end{split}
\end{equation*}
\noindent where $S$ is the truncation order of the Taylor expansion, $t \geq 1$ is a positive integer, and $s!$ is the factorial of $s$.

Due to $\theta_n^\star  \sim \text{Uniform}[0, 2\pi)$, we can obtain
\begin{equation*}
	\begin{aligned}
		\eta = \frac{\gamma^a_P}{\gamma^a_I} &= \frac{\frac{\tau_f^2\tau_g^2 P}{16\sigma^2}
			\left[N^2\pi^2\mathbb{E}[\beta_n]^2 + N\left(16\mathbb{E}[\beta_n^2]-\pi^2\mathbb{E}[\beta_n]^2\right) \right]}{\frac{\tau_f^2\tau_g^2 P}{16\sigma^2}
			\left[N^2\pi^2 + N\left(16-\pi^2\right) \right]}\\
		&=\frac{\pi^2\mathbb{E}[\beta_n]^2 +\frac{\left(16\mathbb{E}[\beta_n^2]-\pi^2\mathbb{E}[\beta_n]^2\right)}{N}}{
			\pi^2 + \frac{16-\pi^2}{N} } \leq 1,
	\end{aligned}
	\vspace{-8pt}
\end{equation*}
\noindent where
\vspace{-8pt}
\begin{equation*}
	\begin{aligned}
		\mathbb{E}[\beta_n] &= \frac{1 - \beta_{\min}}{2^\alpha}\  B + \beta_{\min}, B =  1 + \sum_{k=1}^{\lfloor S/2 \rfloor} \frac{\prod_{t=0}^{2k-1}(\alpha - t)}{2^{2k} (k!)^2},\\
		\mathbb{E}[\beta_n^2] &= \frac{(1-\beta_{\min})^2}{2^{\alpha}} \!\left[1 \! + \! 2B \! + \! C \! + \!  \frac{1-\beta_{\min}}{2^{\alpha-1}} \beta_{\min} B\right]\!\! + \! \beta_{\min}^2, \\
		C &= \sum_{s_1=1}^S \sum_{l_1=0}^{s_1} \sum_{s_2=1}^S \sum_{l_2=0}^{s_2} D_{s_1,l_1} D_{s_2,l_2} \delta_{s_1-2l_1, 2l_2-s_2},\\
		D_{s,l} &= \frac{\prod_{t=0}^{s-1}(\alpha - t)}{s!} \frac{1}{(2j)^s} \frac{s!}{l!(s-l)!} (-1)^l e^{j(2l-s)\delta}.
	\end{aligned}
\end{equation*}

The proof is thus completed.

\vspace{-12pt}
\section*{Appendix C: Proof of Proposition \ref{proposition_2}}
\label{appendices_2}
	Note that objective functions $| \zeta_n - \vartheta_n |^2$ and $| \zeta_m - \vartheta_m |^2$ 
	in problem \eqref{line_search_method1} are independent for $m\neq n$ and the constraints associated with $\theta_n$ and $\theta_m$ are also independent. 
	Therefore, we can solve it by solving the following $N$ independent problems in parallel, which is an element-wise method.
	\vspace{-3pt}
	\begin{equation}\label{line_search_zeta_2}
		\begin{split}
			&\min_{ \theta_n } \left| \rho ( {\rm sin} (\theta_n - \delta) + 1 )^{\alpha} e^{j \theta_n} + \beta_{\min} e^{j \theta_n} - \vartheta_n \right|^2.
		\end{split}
		\vspace{-5pt}
	\end{equation}
	
	The global optimal solution of the problem \eqref{line_search_zeta_2} can be obtained by a line search approach over $\theta_n$, due to the fact that we traverse the entire constraint space. 
	Then, the objective function of problem \eqref{line_search_zeta_2} can be calculated as
	\begin{equation*}
		\begin{split}
			&| \beta(\theta)e^{j\theta} - \vartheta |^2 = |\beta(\theta)|^2 - 2{\rm Re}\{ \beta(\theta)e^{j\theta}\, \vartheta^*  \} + |\vartheta|^2, \\
			&{\rm Re}\{ \beta(\theta)e^{j\theta}\, \vartheta^*  \} = \rho \left({\rm sin}(\theta - \delta) + 1\right)^{\alpha} [ {\rm Re}\{ \vartheta \} {\rm cos}(\theta) +  {\rm Im}\{\vartheta\} \\ 
			&\qquad\qquad\quad\qquad{\rm sin}(\theta) ] + \beta_{\min} [{\rm Re}\{ \vartheta \} {\rm cos}(\theta) + {\rm Im}\{\vartheta\}{\rm sin}(\theta) ], \\
			&|\beta(\theta)|^2 = 
			\rho^2 \left({\rm sin}(\theta \!-\! \delta) \!+\! 1\right)^{2\alpha} \!+\! 2\beta_{\min} \rho \left({\rm sin}(\theta \!-\! \delta) \!+\! 1\right)^{\alpha} \!+\! \beta_{\min}^2.
		\end{split}
	\end{equation*}
	
Therefore, we can reformulate the objective function of problem \eqref{line_search_zeta_2} as
	\begin{equation*}
		\begin{split}
			f(\theta_n|\vartheta_n) &= \rho^2 \left({\rm sin}(\theta_n - \delta) + 1\right)^{2\alpha} + 2\rho ({\rm sin}(\theta_n - \delta) + 1)^{\alpha}  \\
			& [ \beta_{\min} - {\rm Re}\{ \vartheta_n \} {\rm cos}(\theta_n) - {\rm Im}\{\vartheta_n\}{\rm sin}(\theta_n) ]  + \beta_{\min}^2   \\
			& - 2\beta_{\min} [{\rm Re}\{ \vartheta_n \} {\rm cos}(\theta_n) + {\rm Im}\{\vartheta_n\}{\rm sin}(\theta_n) ]  + |\vartheta_n|^2.
		\end{split}
	\end{equation*}
	
	The proof is completed.
\vspace{-10pt}
\bibliographystyle{IEEEtran}
\bibliography{myref}
\end{document}